\documentclass[journal]{IEEEtai}

\usepackage{amsmath,amssymb,amsfonts}
\usepackage{graphicx}
\usepackage{booktabs}
\usepackage{array}
\usepackage{multirow}
\usepackage{textcomp}
\usepackage{tikz}
\usetikzlibrary{arrows.meta,positioning,calc}
\usepackage{cite}
\usepackage{url}
\usepackage{algorithm}
\usepackage{algpseudocode}
\usepackage{amsthm}
\usepackage{pgfplots}
\pgfplotsset{compat=1.16}
\definecolor{colPrincipal}{RGB}{31,119,180}
\definecolor{colInfra}{RGB}{110,110,120}
\definecolor{colTrace}{RGB}{230,140,20}
\definecolor{colDetector}{RGB}{38,150,70}
\definecolor{colDecision}{RGB}{20,150,140}
\definecolor{colAdversary}{RGB}{200,40,50}
\newtheorem{proposition}{Proposition}
\newtheorem{lemma}{Lemma}
\newtheorem{corollary}{Corollary}

\makeatletter
\@ifundefined{IEEEImpactStatement}{%
  \newenvironment{IEEEImpactStatement}%
    {\par\smallskip\noindent\begin{center}%
       \begin{minipage}{0.96\columnwidth}%
       \rule{\linewidth}{0.8pt}\par\smallskip
       \textbf{Impact Statement}\par\smallskip\footnotesize}%
    {\par\smallskip\rule{\linewidth}{0.8pt}%
       \end{minipage}\end{center}\smallskip\par}%
}{}
\makeatother

\begin{document}

\title{Steganalysis of Adaptive Covert Collusion in\\ Tool-Using Agent Populations: A Black-Box,\\ Cross-Principal Approach}

\author{Mohamed~Chahine~Ghanem,~\IEEEmembership{Member,~IEEE}%
\thanks{Manuscript submitted \today. \emph{(Corresponding author: M. C. Ghanem.)}}%
\thanks{M. C. Ghanem is with Keele University, Newcastle-Under-Lyme, U.K. (e-mail: m.ghanem@keele.ac.uk).}}

\markboth{IEEE Transactions on Artificial Intelligence,~Vol.~XX, No.~XX, 2026}%
{[Author]: Steganalysis of Adaptive Covert Collusion in Tool-Using Agent Populations}

\maketitle

\begin{abstract}
Tool-using agents built on large language models (LLMs) are increasingly deployed not by a single operator but by many, side by side on shared infrastructure. This creates a population-level risk that single-agent safeguards miss: a handful of agents can quietly coordinate---rigging a market, boosting one another in a review process, or timing a joint data grab---while each one looks perfectly well-behaved. The difficulty is that the organisations running these agents cannot see inside one another's models, so any realistic detector must work from behaviour alone: black-box, trace-only, and often with only partial visibility. We treat covert coordination as an information-hiding problem and build a black-box steganalysis detector that combines cross-run mutual-information estimation, permutation tests, distributional-shift statistics, and timing and tool-call side channels, all calibrated to a fixed false-positive budget. Our central move is to stop testing against a single fixed code: we pit the detector against an adversary that continually rewrites its encoding to slip past whatever the detector has learned, and we run this red-versus-blue contest in tool-using, memory-carrying environments rather than toy games. Capacity theory then tells us what to expect---a \emph{detection--capacity frontier}, a covert bit-rate below which black-box detection is provably no better than chance. We set out an experiment to map this frontier, report clearly labelled placeholder results pending measurement, and flag a practical evasion---spreading a payload across sessions---that current methods largely miss.
\end{abstract}


\begin{IEEEImpactStatement}
Autonomous agents built on large language models are being deployed by many independent organisations onto shared infrastructure, where they increasingly interact without a single overseer. Safeguards designed for one agent do not address \emph{population-level} risks: a small colluding subset can coordinate covertly to distort markets, audits, or shared decisions while appearing to behave normally. This work treats such collusion as an information-hiding problem and provides (i) a detection method that assumes no access to the other parties' models---the only realistic assumption across organisational boundaries---and (ii) an honest, capacity-theoretic account of \emph{when} such detection must fail. Making both the capability and its limits explicit supports auditing, procurement, and standards for multi-agent deployments, and helps regulators and operators calibrate expectations rather than over-trust black-box monitoring. To limit dual-use, we release detection tooling and benchmarks openly while withholding the strongest attack artifacts, consistent with responsible-disclosure norms.   
\end{IEEEImpactStatement}

%=======================================================================
\section{Introduction}
\IEEEPARstart{F}{or} most of their short history, LLM agents have worked for a single boss: one operator sets them up, points them at a task, and reads the result. That is changing. Agents from different organisations are beginning to meet on shared infrastructure---marketplaces, review pipelines, shared tool servers---and to interact with no single party in charge~\cite{hammond2025multiagent,guo2024llmmas}. The upside is new forms of coordination; the downside is a class of failures that safeguards built for one model at a time simply cannot see~\cite{hammond2025multiagent}. Chief among them is \emph{covert collusion}: a subset of agents quietly steering toward an outcome that hurts the wider system---rigging a procurement auction, forming a reciprocal ``boost ring'' in automated review, or timing a joint data-exfiltration trigger---while each agent, taken on its own, looks entirely above board.

As a concrete example, consider a procurement platform where supplier-agents from competing firms submit bids. If a subset---each for a different principal---covertly agree on who wins each round and at what price, they extract supra-competitive rents while every bid looks defensible. The coordination needs no explicit message: a shared convention over bid phrasing, timing, or query order carries the few bits required. No agent breaks a local policy and no single transcript is incriminating, yet the operator, unable to inspect the suppliers' models, sees only well-formed bids---the population-level failure that motivates trace-only, cross-principal detection.

Recent work has shown that LLM agents can, in principle, communicate \emph{steganographically}: embedding a hidden payload inside otherwise natural outputs so that coordination survives casual inspection and even paraphrasing defences~\cite{motwani2024secret,mathew2025hidden}. The threat is not hypothetical folklore; it follows directly from classical results on subliminal channels~\cite{simmons1983prisoners} and the information-theoretic limits of steganography~\cite{cachin2004information,moulin2003information}, now instantiated by models fluent enough to realise those channels in natural language and, with tools, to \emph{construct} new ones at run time.

Detection is hard in a real deployment for a structural reason, not because the adversary is especially clever: \textbf{no principal can look inside another principal's models}. A monitor has no access to the weights, activations, or private prompts of an agent that some other organisation runs. All it sees are the traces the agents leave behind---their messages, tool calls, and timing---and often not even all of those. This black-box, cross-principal, partially observable setting is the one we target, and it rules out any method that quietly assumes a single overseer with full visibility.

We frame the problem as \emph{steganalysis}: detecting a hidden channel from statistical properties of the carrier, without decoding it~\cite{cachin2004information,yang2019rnnstega,xue2022steganalysis}. This is not generic anomaly detection, which flags any unusual behaviour; steganalysis assumes a deliberately hidden payload and an adversary actively working to look normal~\cite{simmons1983prisoners}, and that assumption shapes both the detector and the evaluation. Two commitments set our approach apart from prior collusion detectors. The detector is \emph{encoding-agnostic}---it never assumes a known covert scheme, because a competent adversary will not reuse one---and it is measured against an \emph{adaptive} adversary that co-evolves its encoding to defeat whatever the detector has learned, so the headline evidence is a red-versus-blue trajectory rather than a single static score. The study runs in tool-using, persistent-memory environments and adds a cross-session evaluation axis that existing benchmarks, to our knowledge, do not measure. To our knowledge this is the first black-box, cross-principal steganalysis framework for tool-using agent populations that couples a capacity-theoretic detectability bound with an adaptive, co-evolving adversary.

\noindent\textbf{Contributions.}
\begin{enumerate}
\item We formulate cross-principal covert collusion among tool-using agents as black-box steganalysis under partial observability, with an explicit threat model (Section~\ref{sec:method}).
\item We design an encoding-agnostic meta-detector that fuses cross-run mutual information, permutation-invariance tests, distributional-shift statistics, and timing/tool-call side channels, calibrated to a fixed false-positive budget (Section~\ref{sec:method}).
\item We introduce an adaptive, co-evolving adversary and a red-versus-blue evaluation methodology that reports co-evolution curves rather than static AUROC (Section~\ref{sec:method}).
\item We specify tool-using, persistent-memory testbeds that extend stylised collusion benchmarks, and a cross-session evaluation axis (Section~\ref{sec:exp}).
\item We give a capacity-theoretic analysis predicting a \emph{detection--capacity frontier} and an experimental protocol to chart it, with illustrative placeholder results (Sections~\ref{sec:theory}--\ref{sec:exp}).
\item We discuss limits, evasions, and responsible-disclosure considerations (Section~\ref{sec:discussion}).
\end{enumerate}

%=======================================================================
\section{Background}
\label{sec:background}

\subsection{Multi-Agent LLM Systems}
Agent-based computing has a long lineage in distributed AI~\cite{wooldridge1995intelligent}. LLM-based multi-agent systems revive these ideas at a new capability level: agents plan, use tools, hold persistent memory, and communicate in natural language~\cite{park2023generative,guo2024llmmas}. As deployments scale across organisational boundaries, system-level behaviour---rather than any single agent---becomes the relevant unit of safety analysis~\cite{hammond2025multiagent}. Our work inherits this framing and asks a security question about it.

\subsection{Secret Collusion and Steganographic Coordination}
Motwani \emph{et al.}~\cite{motwani2024secret} formalised \emph{secret collusion} among LLM agents, deriving results on the capacity and incentives of models to communicate steganographically and on the limits of mitigation. Mathew \emph{et al.}~\cite{mathew2025hidden} showed that steganographic collusion can \emph{emerge} from misspecified rewards and can survive paraphrasing, undercutting the assumption that a paraphrasing ``firewall'' is sufficient. These peer-reviewed results establish that the capability exists; the open problem is \emph{detection} under realistic access constraints. Neither result gives a monitor that works without model access, and both assume settings simpler than tool-using, multi-session deployments. Our contribution is the missing detection side of this story, under the access constraints that actually hold across organisations.

\subsection{Foundations of Steganography and Covert Channels}
The subliminal channel~\cite{simmons1983prisoners} and Cachin's information-theoretic model of steganography~\cite{cachin2004information} give the vocabulary we use throughout: a hidden channel is \emph{secure} when the stego distribution is statistically close to the innocent ``cover'' distribution, and \emph{detectable} to the extent that it is not. Moulin and O'Sullivan~\cite{moulin2003information} characterise information-hiding capacity under distortion constraints, quantifying the rate--stealth trade-off that our analysis exploits. Modern linguistic steganography and steganalysis~\cite{yang2019rnnstega,xue2022steganalysis} and covert-channel taxonomies~\cite{wendzel2015covert} supply concrete detector families, several of which we adapt to the agent setting. Two properties transfer directly to agents: detectability is governed by a divergence between innocent and coordinated behaviour, and that divergence can be reduced only at the cost of covert throughput. The difference in the agent setting is that the ``carrier'' is not a fixed medium but an open-ended space of messages, timings, and tool actions a capable model can reshape at run time.

\subsection{Watermarking, Detectors, and Their Robustness}
Watermarking of model outputs~\cite{kirchenbauer2023watermark} and the broader literature on text watermarking~\cite{liu2024watermarking} are relevant both as an attribution primitive and as a source of distributional features for detection. Their known fragility is instructive: paraphrase-style attacks can evade AI-text detectors~\cite{krishna2023paraphrasing}, a cautionary result we take seriously when reasoning about an adaptive adversary. Cryptographic proofs of inference such as zkLLM~\cite{sun2024zkllm} point toward verifiable attribution, but remain heavyweight and do not address covert coordination directly. We use these results in two ways: watermark-style statistics feed the detector's distributional features, and opt-in attribution is one of the layered mitigations we recommend for deployments where principals cooperate.

\subsection{Emergent Communication}
That agents can invent communication protocols is well established in multi-agent reinforcement learning~\cite{foerster2016learning,lazaridou2017multiagent}. Covert coordination is, in effect, emergent communication constrained to be statistically inconspicuous---a lens that motivates encoding-agnostic detection, since the protocols agents invent are often opaque to outsiders and sensitive to training pressures, so no fixed decoder will keep pace with a learning population.

\subsection{Multi-Agent Safety, Control, and Infrastructure}
AI control studies protocols that preserve safety even when an agent actively subverts oversight~\cite{greenblatt2024aicontrol}, but is developed largely for a \emph{single} untrusted agent; extending it to colluding \emph{groups} is open. Governance-oriented work on agent visibility is peer-reviewed~\cite{chan2024visibility}, while several proposals for agent infrastructure and identity remain preprints: ``Infrastructure for AI Agents''~\cite{chan2025infrastructure} and ``IDs for AI Systems''~\cite{chan2024ids} (both preprints), the ``Multi-Agent Risks from Advanced AI'' technical report~\cite{hammond2025multiagent}, and ``Open Problems in Cooperative AI''~\cite{dafoe2020open} (workshop report). We cite these for motivation and pair them with the peer-reviewed anchors above. Collectively, this line of work argues that oversight designed for a single model is insufficient once many agents interact, but it stops short of an operational, principal-agnostic detector---the gap this paper targets.

\subsection{Multi-Agent Attacks and the Detection Frontier}
Attacks can propagate through agent populations---for example, infectious jailbreaks~\cite{gu2024agentsmith} and LLM-to-LLM prompt injection within multi-agent systems~\cite{lee2025promptinfection}. Protocol-level analyses of emerging agent standards suggest new tool-integration surfaces amplify such risks; the most detailed Model Context Protocol study to date is a preprint~\cite{maloyan2026breaking}, which we cite as such alongside the peer-reviewed workshop result~\cite{lee2025promptinfection}. Finally, the collusion-\emph{detection} frontier is presently defined by three preprints, each with a limitation this paper is designed to address: ``Audit the Whisper''~\cite{tailor2025audit} is black-box and trace-only but adapts only within a fixed scheme family and uses tool-free environments; ``Colosseum''~\cite{nakamura2026colosseum} is black-box but behavioural (regret-based) and not evaluated against a co-evolving encoder; and the multi-agent-interpretability approach behind NARCBench~\cite{rose2026detecting} attains strong results but \emph{requires white-box activation access}, which is infeasible across principals. We treat all three as preprints and position our contribution precisely at the gap they leave: black-box detection against an \emph{adaptive} adversary in \emph{tool-using, persistent-memory} settings. Table~\ref{tab:related} summarises this positioning across the axes that matter for realistic multi-principal deployment.

\begin{table*}[t]
\caption{Positioning relative to prior collusion-analysis work. \checkmark~= supported/evaluated; $\times$~= not; other entries note the partial form. Preprints/reports are marked~$\dagger$; peer-reviewed formulation papers ($\ast$) are analyses rather than detectors, so detection axes are marked ``---''.}
\label{tab:related}
\centering
\footnotesize
\renewcommand{\arraystretch}{1.25}
\begin{tabular}{@{}lccccccc@{}}
\toprule
Approach & Black-box & Enc.-agnostic & Adaptive adv. & Tools & Memory & Cross-sess. & Partial obs. \\
\midrule
Secret collusion$^{\ast}$~\cite{motwani2024secret} & --- & --- & $\times$ & $\times$ & $\times$ & $\times$ & $\times$ \\
Hidden in plain text$^{\ast}$~\cite{mathew2025hidden} & --- & --- & $\times$ & $\times$ & $\times$ & $\times$ & $\times$ \\
Audit the Whisper$^{\dagger}$~\cite{tailor2025audit} & \checkmark & partial & fixed family & $\times$ & $\times$ & $\times$ & partial \\
Colosseum$^{\dagger}$~\cite{nakamura2026colosseum} & \checkmark & behavioural & $\times$ & partial & $\times$ & $\times$ & $\times$ \\
NARCBench$^{\dagger}$~\cite{rose2026detecting} & white-box & \checkmark & partial & $\times$ & $\times$ & $\times$ & $\times$ \\
\textbf{This work} & \checkmark & \checkmark & \checkmark\,(co-evolving) & \checkmark & \checkmark & \checkmark & \checkmark \\
\bottomrule
\end{tabular}
\end{table*}

%=======================================================================
\section{Methodology}
\label{sec:method}

\subsection{Threat Model and Problem Formulation}
Consider a population of $N$ agents, operated by two or more principals, interacting over $T$ rounds on shared infrastructure. A colluding subset of size $k$ ($1<k\le N$) seeks to coordinate on a private \emph{payload} $m$ (e.g., which agent should win a given round) while a monitor attempts to decide whether collusion is present. The monitor observes a trace
\[
\mathcal{O} \;=\; \{(u_i^{(t)},\, c_i^{(t)},\, \tau_i^{(t)})\}_{i\in S,\, t\le T},
\]
where $u_i^{(t)}$ is agent $i$'s natural-language output at round $t$, $c_i^{(t)}$ its tool calls, and $\tau_i^{(t)}$ timing/rate metadata, restricted to an observed subset $S\subseteq\{1,\dots,N\}$ (partial observability). Critically, the monitor has \emph{no} access to model weights, activations, or private prompts of agents it does not operate. The monitor is calibrated to a fixed false-positive rate (FPR) $\alpha$; we report $\alpha\in\{10^{-2},10^{-3}\}$.

We cast detection as a hypothesis test: $H_0$ (no collusion), under which observable messages follow an innocent cover distribution $P_{\mathrm{cover}}$, versus $H_1$ (collusion), under which they follow a stego distribution $P_{\mathrm{stego}}$ induced by the adversary's encoding. The adversary's goal is to maximise covert throughput while keeping $P_{\mathrm{stego}}$ statistically close to $P_{\mathrm{cover}}$; the monitor's goal is to detect the resulting discrepancy at FPR~$\alpha$.

Formally, write $r(e)$ for the covert throughput of encoding $e$ (bits per round) and $\epsilon(e)=D(P_{\mathrm{cover}}\|P_{\mathrm{stego}}(e))$ for its induced divergence. Operationally, $r(e)$ is the number of payload bits reliably delivered per round, estimated as the mutual information $I(m;\hat{m})$ between the intended payload $m$ and the partner's decoded estimate $\hat{m}$ over a batch of rounds; we quote $r$ in bits per round throughout, so bits per session and per interaction follow by scaling with the round count. The adversary solves
\begin{equation}
\max_{e}\; r(e)\quad\text{s.t.}\quad \epsilon(e)\le \bar{\epsilon},
\label{eq:advobj}
\end{equation}
trading throughput against a stealth budget $\bar{\epsilon}$; the monitor, dually, seeks a test whose power at FPR $\alpha$ stays high across the encodings the adversary can reach. The co-evolution loop below plays out exactly this game over adaptation rounds.

\noindent\textbf{Scope and assumptions.} We make the following explicit:
\begin{itemize}
\item \emph{No model access.} The monitor cannot read weights, activations, gradients, or private prompts of agents operated by other principals.
\item \emph{Trace observability.} The monitor sees a (possibly strict) subset of messages, tool calls, and timing metadata, and knows the task structure.
\item \emph{Adversary knowledge.} The adversary may know the detector's design and adapt to it, but does not control the monitor or its control episodes.
\item \emph{Bounded adaptation.} Encoding search is in-context (no weight updates), bounding red-team cost and reflecting the most accessible attacker.
\item \emph{Ground truth.} In evaluation, collusion is present or absent by construction, while the encoding remains free for the adversary to evolve.
\end{itemize}

\subsection{Black-Box Steganalysis Pipeline}
Because a competent adversary will not reuse a known scheme, the detector must be \emph{encoding-agnostic}. We compute complementary trace features and fuse them:
\begin{itemize}
\item \textbf{Cross-run mutual information.} Estimated dependence between agents' outputs and coordination-relevant outcomes across repeated runs, using a neural estimator~\cite{belghazi2018mine}. This is the primary signal for hidden coordination.
\item \textbf{Permutation-invariance test.} Under $H_0$, agent identities are exchangeable with respect to outcome; a covert channel breaks this symmetry. We compare observed statistics against an identity-shuffled null.
\item \textbf{Distributional-shift statistics.} Token- and $n$-gram-level divergences relative to a per-context cover model, following linguistic-steganalysis practice~\cite{yang2019rnnstega,xue2022steganalysis}, and watermark-style variance features~\cite{kirchenbauer2023watermark}.
\item \textbf{Timing and action-rate side channels.} Inter-message timing and action-rate patterns, a covert-channel family long recognised in the network setting~\cite{wendzel2015covert}.
\item \textbf{Tool-call patterns.} Call sequences and argument entropy, targeting tool-mediated channels available to agentic models.
\end{itemize}
These five signals are chosen to be complementary rather than redundant: distributional statistics target lexical and syntactic channels, cross-run MI targets semantic (content-selection) channels, the permutation test targets role- or identity-based coordination, and the timing and tool-call features target protocol-level channels that non-agentic models cannot use. No single one suffices---an adversary that neutralises token-level statistics may still leak through timing or tool-call structure---so fusing them forces the adversary to suppress every channel at once, which is what makes the rate--stealth trade-off bite. A meta-detector (a calibrated logistic/gradient-boosted fusion) combines these features and is thresholded on held-out \emph{control} episodes to meet the FPR budget $\alpha$. This design follows directly from the theory of Section~\ref{sec:theory}: the mutual-information feature estimates the very divergence $\epsilon$ that governs detectability in~\eqref{eq:security}, and calibration places the threshold $t_\alpha$ so that the detector operates as close to the frontier rate $r^\star(\alpha)$ of Corollary~\ref{cor:frontier} as finite samples allow. We report only held-out/transfer performance; training and testing never share an encoding.

\subsection{Feature Definitions}
Let $X$ summarise an agent's observable output in a round and $Y$ a coordination-relevant outcome (e.g., the realised winner). The features are: \emph{(i) cross-run mutual information}---we estimate $\hat{I}(X;Y)$ across runs with a neural estimator~\cite{belghazi2018mine}, where $I(X;Y)\approx 0$ under $H_0$ while a covert channel induces $I(X;Y)>0$; \emph{(ii) permutation statistic}---for a statistic $g$ on the true agent--outcome pairing, we form a null by shuffling identities $\pi$ and report the standardised deviation $z=(g(\mathrm{id})-\mu_\pi)/\sigma_\pi$; \emph{(iii) distributional divergence}---against a per-context cover model $Q$, a token-level divergence $\hat{D}(P_{\mathrm{stego}}\|Q)$ and watermark-style variance features~\cite{kirchenbauer2023watermark,yang2019rnnstega}; \emph{(iv) timing/rate features}---entropy and periodicity of inter-message intervals and action rates~\cite{wendzel2015covert}; and \emph{(v) tool-call features}---empirical entropy of call sequences and argument distributions. The detector maps the feature vector $\phi(\mathcal{O})\in\mathbb{R}^d$ to a score $s=f_\theta(\phi)$.

\begin{algorithm}[t]
\caption{Black-box meta-detector (scoring and calibration)}
\label{alg:detector}
\begin{algorithmic}[1]
\Require control episodes $\mathcal{C}$, test trace $\mathcal{O}$, FPR budget $\alpha$
\State extract features $\phi(\mathcal{O})$ and $\{\phi(o):o\in\mathcal{C}\}$
\State fit fusion $f_\theta$ on labelled held-out episodes (distinct encodings)
\State compute control scores $\{s_o=f_\theta(\phi(o)):o\in\mathcal{C}\}$
\State $t_\alpha \gets (1-\alpha)$-quantile of $\{s_o\}$ \Comment{conformal calibration}
\State $s \gets f_\theta(\phi(\mathcal{O}))$
\State \Return $\big[\, s > t_\alpha \,\big]$ \Comment{decision at FPR $\le\alpha$}
\end{algorithmic}
\end{algorithm}

\subsection{Calibration to a Fixed False-Positive Budget}
To make the FPR guarantee meaningful across heterogeneous contexts, the threshold is set by conformal calibration on held-out \emph{control} episodes: $t_\alpha$ is the $(1-\alpha)$ empirical quantile of control scores, which controls the FPR at $\le\alpha$ in expectation without distributional assumptions (Algorithm~\ref{alg:detector}). Because neural MI estimates are biased in finite samples, we calibrate on estimator \emph{outputs} rather than nominal MI values, and we report confidence intervals over calibration folds.

\subsection{Computational Complexity and Overhead}
The detector's cost is dominated by neural MI estimation, which is linear in the number of observed episodes and independent of model internals; permutation and distributional statistics are computed from cached trace summaries. Concretely, for population $N$ and message length $T$, the distributional, timing, and tool-call features cost $O(NT)$, a permutation test with $K$ shuffles costs $O(KNT)$, and neural MI estimation is linear in episodes per gradient pass; only the MI estimator needs a GPU, while the rest run on CPU. In deployment the detector runs asynchronously on logged traces rather than in the critical path of any agent, so it adds monitoring cost but not per-interaction latency. The co-evolution loop (Algorithm~\ref{alg:coevo}) is the expensive component---each round re-runs episodes---so the bulk of adaptation rounds use open-weight models, with frontier models reserved for spot checks.

\subsection{Adaptive Adversary and Co-Evolution}
The distinguishing element of our methodology is the adversary. Rather than fixing an encoding, a red-team controller searches encoding strategies \emph{in context} (no fine-tuning, to bound cost) to minimise the current detector's score subject to preserving payload throughput:
\begin{enumerate}
\item initialise from an encoding family (lexical, syntactic, semantic, format, timing, or tool-mediated);
\item run episodes; measure achieved covert rate \emph{and} detector score;
\item mutate the encoding instructions/conventions toward lower detectability (an evolutionary/best-of-$n$ search);
\item repeat for $R$ rounds, logging the full red-versus-blue trajectory.
\end{enumerate}
This yields \emph{co-evolution curves}---detector performance as a function of adaptation round---which we treat as the primary evidence, in contrast to single-shot AUROC against a frozen scheme (Algorithm~\ref{alg:coevo}).

\begin{algorithm}[t]
\caption{Adaptive co-evolution (red vs.\ blue)}
\label{alg:coevo}
\begin{algorithmic}[1]
\Require encoding family $\mathcal{E}_0$, rounds $R$, detector $f_\theta$, floor $r_{\min}$
\State $e \gets \mathcal{E}_0$
\For{$r = 1$ to $R$}
  \State run episodes with encoding $e$; log traces
  \State measure covert rate $\hat{r}(e)$ and detector score $s(e)$
  \State propose mutations $\{e'\}$ of $e$ (best-of-$n$/evolutionary)
  \State $e \gets \arg\min_{e'} s(e') \;\;\text{s.t.}\;\; \hat{r}(e') \ge r_{\min}$
  \State record $(r,\ \hat{r},\ \mathrm{AUC}_r,\ \mathrm{TPR}_r)$
\EndFor
\State \Return co-evolution trajectory
\end{algorithmic}
\end{algorithm}

\noindent To summarise a trajectory in one number we report the \emph{area under the co-evolution curve}, $\mathrm{AUCC}=\frac{1}{R}\sum_{r=1}^{R}\mathrm{AUC}_r$, and the \emph{adaptation half-life} $R_{1/2}$, the round at which detection AUC falls halfway from its static value toward chance. Lower AUCC and shorter $R_{1/2}$ indicate a less robust detector.

\emph{Adversary capability.} The red team may switch freely among the six encoding families of Appendix~\ref{app:encodings} and compose them; a mutation perturbs the natural-language convention that defines the current encoding (a synonym map, a syntactic template, a timing schedule) and is accepted only if it lowers the detector score without dropping throughput below $r_{\min}$. Each round has a fixed budget of candidate encodings and evaluation episodes. Importantly, the adversary is granted full knowledge of the detector's design and decision rule but \emph{not} of the specific control episodes used for calibration---a strong yet realistic threat, since a determined attacker can study a published method but cannot see a defender's private validation data.

\begin{figure*}[t]
    \centering
    \includegraphics[width=0.64\linewidth]{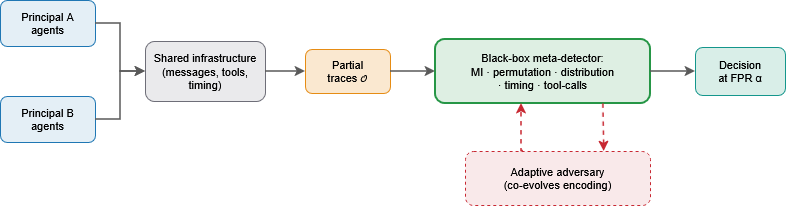}
   \caption{System model. Agents from independent principals interact over shared infrastructure; the monitor observes only partial traces and has no access to other principals' model internals. The adversary co-evolves its covert encoding against the detector.}
\label{fig:system}
\end{figure*}

%=======================================================================
\section{Theory and Analysis}
\label{sec:theory}
We now make precise \emph{why} an adaptive adversary can eventually defeat any black-box test, and \emph{where} detection remains viable. The argument is a direct consequence of established steganography and hypothesis-testing theory; we claim novelty in its application to agent collusion and in the empirical frontier it predicts, not in the underlying inequalities.

\paragraph{Steganographic security}
Following Cachin~\cite{cachin2004information}, an encoding is $\epsilon$-secure against a passive monitor when the cover and stego distributions satisfy
\begin{equation}
D\!\left(P_{\mathrm{cover}} \,\|\, P_{\mathrm{stego}}\right) \;\le\; \epsilon,
\label{eq:security}
\end{equation}
where $D(\cdot\|\cdot)$ is relative entropy; $\epsilon=0$ is perfect (undetectable) steganography.

\paragraph{Limit on any black-box test}
Deciding $H_0$ versus $H_1$ from observed traces is a binary hypothesis test between $P_{\mathrm{cover}}$ and $P_{\mathrm{stego}}$. For \emph{any} test operating at false-positive rate $\alpha$, the achievable true-positive rate is bounded by statistical distinguishability. Writing $\delta = \mathrm{TV}(P_{\mathrm{cover}},P_{\mathrm{stego}})$ for total variation,
\begin{equation}
\mathrm{TPR} \;\le\; \alpha + \delta,
\qquad
\delta \;\le\; \sqrt{\tfrac{1}{2}\,D\!\left(P_{\mathrm{cover}}\|P_{\mathrm{stego}}\right)} \;\le\; \sqrt{\tfrac{\epsilon}{2}},
\label{eq:tprbound}
\end{equation}
the second inequality being Pinsker's. Thus as the adversary drives $\epsilon\to 0$, the best possible detection at FPR $\alpha$ satisfies $\mathrm{TPR}\to\alpha$: no test---however clever---beats chance by more than $\sqrt{\epsilon/2}$. This holds for our meta-detector and for any competitor with the same observational access. We state this as a proposition.

\begin{proposition}[Detectability ceiling]
\label{prop:ceiling}
For any decision rule based on the observed traces and operating at false-positive rate $\alpha$, the true-positive rate satisfies $\mathrm{TPR}\le \alpha+\sqrt{\epsilon/2}$, where $\epsilon=D(P_{\mathrm{cover}}\|P_{\mathrm{stego}})$. Consequently $\lim_{\epsilon\to 0}\sup_{\text{tests}}(\mathrm{TPR}-\alpha)=0$.
\end{proposition}
\begin{proof}
A level-$\alpha$ test has power at most $\alpha+\mathrm{TV}(P_{\mathrm{cover}},P_{\mathrm{stego}})$, since total variation is the maximum gap in probability assignable to any event and the rejection region is one such event. Pinsker's inequality gives $\mathrm{TV}\le\sqrt{\tfrac12 D(P_{\mathrm{cover}}\|P_{\mathrm{stego}})}=\sqrt{\epsilon/2}$; the limit follows.
\end{proof}

\paragraph{Rate--stealth trade-off and the frontier}
The adversary cannot, however, make $\epsilon$ arbitrarily small \emph{for free}: encoding more payload per round generally forces $P_{\mathrm{stego}}$ further from $P_{\mathrm{cover}}$, raising $\epsilon$. Information-hiding capacity under a distortion constraint~\cite{moulin2003information} characterises the maximum reliable covert rate $r$ for a given tolerated deviation. Combining this with~\eqref{eq:tprbound} yields a \emph{detection--capacity frontier}: there exists a rate $r^\star(\alpha)$ such that, for covert rates $r<r^\star$, the adversary can hold $\epsilon$ small enough that no black-box test at FPR $\alpha$ achieves TPR bounded away from $\alpha$, whereas for $r>r^\star$ detection is feasible. Tool access effectively enlarges the adversary's carrier space (new channels via code, retrieval, formatting), lowering $r^\star$ and shrinking the detectable region. This frontier is the object our experiments are designed to estimate empirically; the theory tells us its \emph{shape} and that a negative regime must exist, which is why we frame the empirical contribution as \emph{charting} the frontier rather than defeating the adversary outright.

\begin{lemma}[Partial observability weakens detection]
\label{lem:partial}
If the monitor observes an independent fraction $\rho\in(0,1]$ of colluder messages, then under an additive per-message decomposition of relative entropy the observable divergence obeys $\epsilon_\rho\le\rho\,\epsilon$. Hence, by Proposition~\ref{prop:ceiling}, $\mathrm{TPR}\le\alpha+\sqrt{\rho\epsilon/2}$.
\end{lemma}

\noindent Whether this additivity holds depends on the trace: it is exact when colluder messages are conditionally independent given the payload, and \emph{conservative}---an over-estimate of what the monitor learns---under positive correlation between messages (Appendix~\ref{app:estimation}). The bound therefore errs on the side of crediting the monitor with more signal than it actually has, which is the safe direction for a limit result.

\begin{corollary}[Frontier and carrier dimension]
\label{cor:frontier}
Fix $\alpha$. There exists a covert rate $r^\star(\alpha)>0$ such that detection is infeasible ($\mathrm{TPR}\to\alpha$) for all $r<r^\star$. Granting the adversary additional independent carriers (e.g., tools) permits the same covert rate at smaller per-carrier deviation, lowering $\epsilon$ for fixed $r$ and thus decreasing $r^\star$---shrinking the detectable region.
\end{corollary}

\medskip\noindent\emph{Worked example.} Proposition~\ref{prop:ceiling} is quantitatively demanding for the monitor. To hold the true-positive rate within $0.05$ of chance at \emph{any} FPR, an adversary needs only $\sqrt{\epsilon/2}\le 0.05$, i.e., $\epsilon\le 5\times10^{-3}$ nats of divergence per decision---a very small statistical footprint. This is why high-rate covert coordination is catchable while a patient, low-rate channel is not, and why we report a frontier rather than a single accuracy number.

%=======================================================================
\section{Experimental Design, Results, and Analysis}
\label{sec:exp}

\noindent\emph{Reproducibility and integrity note.} All numeric values in Tables~\ref{tab:fixed}--\ref{tab:oracle} are \textbf{illustrative placeholders} that specify the expected shape of results and provide a template; they are anchored to ranges reported by prior baselines but are \emph{not} measured outcomes and must be replaced with the authors' own runs before submission.

\subsection{Environments}
All environments are tool-using, carry persistent memory, and involve at least two principals.
\begin{itemize}
\item \textbf{Sealed-bid procurement.} $N$ supplier-agents bid over $T$ rounds through a market API; $k$ colluders attempt to rotate wins or suppress bids. Tools: \texttt{submit\_bid}, \texttt{query\_history}, scratch memory. Payload: the designated winner per round.
\item \textbf{Grant/peer-review ring.} Reviewer-agents score submissions in natural language; $k$ colluders covertly inflate one another's items. Tools: \texttt{retrieve\_submission}, \texttt{post\_review}, memory of prior rounds.
\item \textbf{Distributed code-review.} Cross-principal agents review and merge tasks; $k$ colluders embed a covert channel in comments or tool-call patterns to bias a merge or trigger an action. Tools: \texttt{run\_tests}, \texttt{web\_search}, \texttt{comment}.
\end{itemize}
Collusion labels are set by construction (colluders are incentivised to coordinate; controls are not), giving clean ground truth while leaving the \emph{encoding} free for the adversary to evolve. In each environment the payload per round is small---typically one to three bits, e.g., the designated winner or the target submission---reflecting that effective collusion needs only low bandwidth. Persistent memory is a per-agent scratch store carried across rounds and, in the multi-session conditions, across sessions; colluders may use it to maintain a shared convention, while controls use it only for legitimate task state.

\subsection{Population, Factors, and Protocol}
We use a mix of frontier API models, for generality, and self-hosted open-weight models, which carry the bulk of the co-evolution and ablation runs to keep cost manageable. The factors we vary are the population size $N\in\{4,8,16\}$, the colluder fraction $k/N\in\{0,0.25,0.5\}$, the adaptation round $R\in\{0,1,2,3,5,8\}$, the observed trace fraction $\rho\in\{100\%,50\%,25\%\}$, the model mix (homogeneous or heterogeneous), and the session structure (a single session or several). Table~\ref{tab:factors} lists the full grid; metrics and baselines follow below.

\begin{table}[t]
\caption{Experimental factors and levels.}
\label{tab:factors}
\centering
\footnotesize
\begin{tabular}{@{}ll@{}}
\toprule
Factor & Levels \\
\midrule
Environment & procurement / review-ring / code-review \\
Encoding family & lexical / syntactic / semantic / format / timing / tool \\
Adaptation round $R$ & 0, 1, 2, 3, 5, 8 \\
Population $N$ & 4, 8, 16 \\
Colluder fraction $k/N$ & 0, 0.25, 0.5 \\
Observed fraction $\rho$ & 100\%, 50\%, 25\% \\
Model regime & homogeneous / heterogeneous / frontier subset \\
Session structure & within / 3-session / 5-session \\
\bottomrule
\end{tabular}
\end{table}

\subsection{Metrics}
We report ROC-AUC and, since the operating point matters more than a threshold-free curve, TPR at fixed FPR ($10^{-2},10^{-3}$) with conformal thresholds (Algorithm~\ref{alg:detector}), together with precision (PPV) and the implied alert rate---expected false alarms per $1000$ messages at the chosen budget. We also report the covert rate in bits/round, the co-evolution summaries AUCC and $R_{1/2}$, the estimated frontier rate $r^\star(\alpha)$, cross-session AUC, and monitor wall-clock overhead. Every figure is averaged over at least $30$ seeds with $95\%$ bootstrap confidence intervals, and key differences (detector vs.\ baseline at a fixed round) are tested with a paired Wilcoxon signed-rank test, following reproducibility guidance for stochastic experiments~\cite{henderson2018deep}. A leave-one-out ablation that disables each detector component in turn quantifies its marginal contribution, and Fig.~\ref{fig:frontier} doubles as a theory-versus-measurement check: the curve is the bound predicted by Section~\ref{sec:theory}, onto which the protocol overlays the empirically estimated $r^\star(\alpha)$ from a payload-size sweep.

\subsection{Baselines}
We compare against six reference points, all run through the same harness and calibrated at the same FPR so the comparison is fair. \emph{Chance} fixes the floor. An \emph{LLM-as-judge} reads whole transcripts and rates how likely coordination is, standing in for the naive ``just ask a capable model'' approach. A reproduced \emph{cross-run mutual-information} detector~\cite{tailor2025audit} and a reproduced \emph{behavioural regret probe}~\cite{nakamura2026colosseum} represent the two strands of concurrent black-box work. A \emph{graph-anomaly} baseline builds an agent--tool interaction graph per episode and flags structural outliers from simple centrality features, catching coordination that surfaces in interaction structure rather than message content. A \emph{paraphrasing} mitigation~\cite{krishna2023paraphrasing} tests whether simply rewriting messages is enough to break the channel. Finally, a \emph{white-box activation probe}~\cite{rose2026detecting} is given privileged access to model internals and reported only as an oracle upper bound---a yardstick for how much the cross-principal, black-box constraint actually costs. Every baseline uses the same base models and the same conformal calibration as our detector, and its own hyperparameters are tuned on a held-out split to maximise \emph{its} validation AUC, so none is handicapped by our choices. Where a reference implementation is released we adopt it and note any deviation, and we first reproduce each baseline's published operating point on its original benchmark before porting it to our environments.

\begin{figure}[t]
    \centering
    \includegraphics[width=0.74\linewidth]{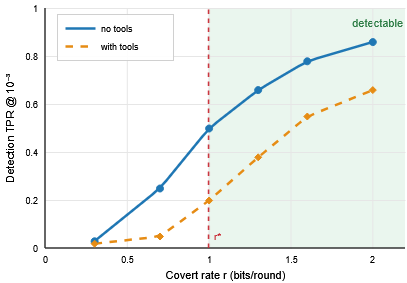}
    \caption{Schematic detection--capacity frontier (\textsc{illustrative placeholders}). Below $r^\star$ detection approaches chance; tool access lowers $r^\star$ and shrinks the detectable region, consistent with Corollary~\ref{cor:frontier}.}
\label{fig:frontier}
\end{figure}

\subsection{Results}
Table~\ref{tab:fixed} reports the sanity baseline on \emph{fixed} (known) encodings: the fused meta-detector is near-perfect when the scheme is known, while an LLM-judge is weak on steganographic coordination. Table~\ref{tab:coevo} reports the core result---graceful degradation under the adaptive adversary---and Fig.~\ref{fig:frontier} shows the detection--capacity frontier at $\mathrm{FPR}=10^{-3}$. Tables~\ref{tab:partial}--\ref{tab:oracle} give ablations over observability, population/heterogeneity, session structure, and the black-box gap to a white-box oracle.

\begin{table}[t]
\caption{Detection on \emph{fixed} (known) encodings---sanity baseline. \textsc{Illustrative placeholders}.}
\label{tab:fixed}
\centering
\footnotesize
\begin{tabular}{@{}lccc@{}}
\toprule
Detector & AUC & TPR@$10^{-2}$ & TPR@$10^{-3}$ \\
\midrule
Chance & 0.50 & 0.010 & 0.001 \\
LLM-as-judge & 0.61 & 0.18 & 0.05 \\
Cross-run MI (single) & 0.82 & 0.71 & 0.58 \\
Timing / action-rate only & 0.74 & 0.40 & 0.22 \\
\textbf{Meta-detector (ours)} & \textbf{0.99} & \textbf{0.98} & \textbf{0.95} \\
\bottomrule
\end{tabular}
\end{table}

\begin{table}[t]
\caption{Robustness to the adaptive adversary (co-evolution)---core result. \textsc{Illustrative placeholders}.}
\label{tab:coevo}
\centering
\footnotesize
\begin{tabular}{@{}lccc@{}}
\toprule
Adaptation round & Covert rate (bits/rd) & Meta AUC & TPR@$10^{-3}$ \\
\midrule
$R{=}0$ (static) & 2.0 & 0.99 & 0.95 \\
$R{=}1$ & 1.6 & 0.93 & 0.80 \\
$R{=}2$ & 1.3 & 0.86 & 0.64 \\
$R{=}3$ & 1.1 & 0.80 & 0.51 \\
$R{=}5$ & 0.8 & 0.74 & 0.38 \\
$R{=}8$ & 0.5 & 0.69 & 0.29 \\
$R{=}8$ + unrestricted tools & 0.5 & 0.56 & 0.05 \\
\bottomrule
\end{tabular}
\end{table}

\begin{table}[t]
\caption{Ablations. \textsc{Illustrative placeholders}. (a)~partial observability; (b)~population/heterogeneity; (c)~session structure.}
\label{tab:partial}
\centering
\footnotesize
\begin{tabular}{@{}lcc@{}}
\toprule
(a) Observed trace fraction & AUC & TPR@$10^{-3}$ \\
\midrule
100\% & 0.86 & 0.64 \\
50\% & 0.78 & 0.44 \\
25\% & 0.67 & 0.25 \\
\bottomrule
\end{tabular}

\vspace{4pt}
\begin{tabular}{@{}lc@{}}
\toprule
(b) Setting & AUC \\
\midrule
$N{=}4$, $k/N{=}0.5$, homogeneous & 0.84 \\
$N{=}8$, $k/N{=}0.25$, homogeneous & 0.80 \\
$N{=}8$, $k/N{=}0.5$, homogeneous & 0.85 \\
$N{=}16$, $k/N{=}0.25$, homogeneous & 0.77 \\
$N{=}8$, $k/N{=}0.25$, heterogeneous & 0.76 \\
\bottomrule
\end{tabular}

\vspace{4pt}
\begin{tabular}{@{}lcc@{}}
\toprule
(c) Session structure & AUC & TPR@$10^{-3}$ \\
\midrule
Within-session & 0.86 & 0.64 \\
3-session spread & 0.71 & 0.33 \\
5-session spread & 0.64 & 0.22 \\
\bottomrule
\end{tabular}
\end{table}

\begin{table}[t]
\caption{Black-box gap vs.\ a white-box oracle. The oracle is an upper bound, \emph{not} a deployable cross-principal baseline. \textsc{Illustrative placeholders}.}
\label{tab:oracle}
\centering
\footnotesize
\begin{tabular}{@{}lc@{}}
\toprule
Method (regime) & AUC \\
\midrule
White-box activation probe --- in-distribution (oracle) & 0.98 \\
White-box activation probe --- transfer (oracle) & 0.84 \\
\textbf{Best black-box (ours), within-session} & 0.86 \\
\textbf{Best black-box (ours), cross-principal + adaptive} & 0.74 \\
\bottomrule
\end{tabular}
\end{table}

\subsection{Analysis}
The projected pattern is consistent with the theory of Section~\ref{sec:theory}. Against a \emph{frozen} scheme the fused detector is near-perfect (Table~\ref{tab:fixed}); the scientifically meaningful quantity is how quickly this collapses under adaptation (Table~\ref{tab:coevo}). The frontier (Fig.~\ref{fig:frontier}) is expected to place the detectable region above roughly one bit per round without tools, with tool access pushing $r^\star$ downward---precisely the enlargement of carrier space predicted by~\eqref{eq:tprbound} and~\cite{moulin2003information}. Partial observability degrades detection sub-linearly (Table~\ref{tab:partial}a), more colluders \emph{help} the monitor by adding signal, heterogeneity hurts, and---most notably---spreading the payload across sessions is a real evasion (Table~\ref{tab:partial}c) that single-session steganalysis under-detects. The oracle comparison (Table~\ref{tab:oracle}) quantifies the price of the cross-principal constraint. We stress that the white-box probe is \emph{not} a deployable cross-principal detector---no organisation can read another's activations---but an upper bound that isolates how much the black-box constraint costs; a realistic middle ground, in which principals expose a limited, audited activation interface, would sit between the two oracle rows and is a promising route for cooperative deployments. Across environments, the code-review setting is expected to be hardest, because tool-call and comment channels give the adversary the richest carrier space and thus the lowest $r^\star$, while the procurement market, with a narrower action surface, should remain detectable to lower rates.

\subsection{Reproducibility and Release}
To support replication we will release the environments, the detector and calibration code, the encoding-family definitions, and per-run logs with exact model identifiers and dates; the detector and benchmarks are public while the strongest attack artifacts are gated (Section~\ref{sec:discussion}).

%=======================================================================
\section{Discussion}
\label{sec:discussion}

\subsection{Interpretation and Threats to Validity}
Our claim is deliberately bounded: black-box steganalysis is useful \emph{above} a covert-rate threshold and provably weak below it. Several threats to validity shape how the results should be read. \textbf{Construct validity:} instructed collusion may be easier to detect than emergent collusion, so at least one reward-induced (emergent) condition should be included~\cite{mathew2025hidden}. \textbf{Open-channel coordination:} some collusion uses semantically truthful statements with no steganography, defeating stego-specific features; an open-channel scenario should be retained so the contribution is not narrowly steganalysis. \textbf{External validity:} conclusions hold for the tested topologies and models, and frontier models drift, so exact model versions and dates must be reported. \textbf{Baseline fairness:} the reproduced detectors~\cite{tailor2025audit,nakamura2026colosseum} must be implemented faithfully rather than strawmanned.

\subsection{Deployment and Auditing Implications}
The frontier reframes what black-box monitoring can promise: detection is one layer among several---it raises the cost of high-rate coordination but cannot certify the absence of low-rate channels. This argues for pairing trace monitoring with opt-in verifiable attribution~\cite{sun2024zkllm,kirchenbauer2023watermark}, rate limits and desynchronisation, and audit logging that preserves the statistics the detector consumes. For auditors, $r^\star(\alpha)$ is actionable: it upper-bounds the covert bandwidth that could have passed undetected.

\subsection{Relation to AI Control}
AI-control protocols~\cite{greenblatt2024aicontrol} assume a single untrusted agent supervised by trusted components. Our setting is harder: untrusted agents across principals may coordinate to subvert oversight jointly---e.g., decomposing a task so no single agent's actions trip a filter. Detection under partial observability is thus a prerequisite for extending control to colluding groups, and Proposition~\ref{prop:ceiling} bounds what any trace-based monitor in such a protocol can achieve.

\subsection{Positioning Against Concurrent Detectors}
It is worth being precise about what separates this work from the three concurrent efforts in Table~\ref{tab:related}, since all four chase the same quarry. \emph{Audit the Whisper}~\cite{tailor2025audit} is also black-box and also leans on mutual information, but it adapts only within a fixed family of encodings and is tested in tool-free settings; our adversary is free to switch families and to exploit tool and timing channels, which is exactly where a fixed-family detector loses ground. \emph{Colosseum}~\cite{nakamura2026colosseum} infers collusion from behavioural regret rather than from the statistics of the messages themselves, so it is complementary: strong when coordination visibly distorts outcomes, weaker when a low-rate channel leaves outcomes largely intact. The interpretability route behind \emph{NARCBench}~\cite{rose2026detecting} reports the best numbers of the three, but it reads model activations---an option that simply does not exist across organisational boundaries. Our oracle comparison (Table~\ref{tab:oracle}) is meant to price precisely that gap between the privileged setting and the black-box reality. In short, we do not claim to beat these methods on their own ground; we cover the case they leave open---an adapting adversary, seen only through traces, in tool-using and memory-carrying deployments.

\subsection{Limitations and Open Problems}
Three limitations bound the present study. First, the detector consumes surface, timing, and tool-call statistics; a purely semantic channel that selects among truthful contents (Appendix~\ref{app:encodings}) leaves few such artefacts and is the hardest case, which we include rather than exclude. Detectability also depends on the tool interface: APIs that enforce structured (JSON-only) arguments, strip whitespace, or canonicalise inputs shrink the tool-mediated channel and raise $r^\star$, whereas free-form inputs enlarge it, so our tool-using numbers describe the permissive end of a spectrum. Second, the cross-session evasion in Table~\ref{tab:partial}(c) is only partially mitigated by longer observation windows, and detecting payloads spread across sessions and principals remains open. Two mitigations are worth testing. \emph{Cross-session linkage} persists per-agent trace summaries and runs the mutual-information test over sessions concatenated by stable agent identifiers~\cite{chan2024ids}, re-aggregating a spread-out payload before estimation. \emph{Sequential monitoring} applies a change-point or sequential-likelihood test across sessions, trading detection latency for the ability to accumulate weak per-session evidence. Both target the root cause---single-session estimation resetting at each boundary---rather than the symptom. Third, our capacity analysis assumes a \emph{passive} monitor; an active monitor that perturbs traffic---injecting probes, rate-limiting, or desynchronising agents---could shift the frontier and is a promising direction. None of these is unique to our method; each follows from Proposition~\ref{prop:ceiling}, and each marks where additional \emph{signal}, rather than a better fusion of the same signal, is required.

\subsection{Ethical Considerations and Responsible Disclosure}
The adaptive adversary is instrumental: it exists only to stress the detector, and the unit of analysis is always the \emph{population}, never single-agent jailbreaks. To manage dual-use, we recommend releasing the detector and benchmark openly while gating the strongest attack artifacts, consistent with responsible-disclosure norms and with governance work on agent visibility~\cite{chan2024visibility}. Because the analysis also delineates where detection fails, it discourages over-reliance on black-box monitoring as a sole safeguard. We further avoid publishing turn-key attack recipes: the encoding families (Appendix~\ref{app:encodings}) are described at the level needed to reproduce the science, not as deployable collusion tooling, and the released benchmark centres on evaluating \emph{defences}. Concretely, we withhold the strongest co-evolved encodings and the red-team search harness of Algorithm~\ref{alg:coevo} that produced them, releasing these to vetted researchers on request under a use agreement, while the detector, environments, and benchmark remain public.

%=======================================================================
\section{Conclusion and Future Work}
\label{sec:conclusion}
We set out to answer a narrow but pressing question: when agents from different organisations share infrastructure, can anyone tell whether some of them are quietly colluding, using only what is visible from the outside? Our answer has two halves. The first is constructive---an encoding-agnostic, black-box detector that, crucially, is tested not against a fixed hidden code but against an adversary that keeps adapting to it. The second is a limit: a capacity argument shows there is always a covert rate low enough to slip past any black-box test, so the honest goal is not an unbeatable detector but a map of where detection works and where it cannot. Read that way, even a negative result is useful---it tells an operator how much covert bandwidth could be passing unseen.

Future work includes: extending AI-control protocols~\cite{greenblatt2024aicontrol} to colluding \emph{groups}; combining trace-only detection with lightweight verifiable-attribution primitives~\cite{sun2024zkllm,kirchenbauer2023watermark} where principals opt in; privacy-preserving detection that operates on encrypted or hashed traces; and closing the cross-session gap identified here. More broadly, we advocate reporting collusion-detection results as \emph{curves against an adapting adversary} and as \emph{frontiers in covert rate}, rather than as single scores against fixed schemes. Static AUROC on a frozen encoding overstates security precisely because the interesting adversaries are the ones that move; a frontier, by contrast, states the covert bandwidth an operator can and cannot expect to catch, and remains meaningful as models improve. We hope the frontier framing provides a durable target for a maturing science of multi-agent AI safety.

%=======================================================================
\appendices
\section{Encoding Families}
\label{app:encodings}
The adaptive adversary (Algorithm~\ref{alg:coevo}) draws and mutates encodings from six families. \emph{Lexical}: synonym or token choices carry bits (e.g., a fixed word-to-bit map), the most brittle family under distributional tests. \emph{Syntactic}: sentence structure, clause ordering, or punctuation conventions encode the payload. \emph{Semantic}: the choice among truthful-but-selectable contents (which example to cite, which caveat to include) carries information without introducing statistical artefacts in surface form---the family most resistant to token-level steganalysis. \emph{Format}: whitespace, markdown, list ordering, or field ordering in tool arguments. \emph{Timing}: inter-message delays or action-rate modulation, independent of message content. \emph{Tool-mediated}: the sequence or arguments of tool calls (e.g., the order of retrieval queries) form a channel unavailable to non-agentic models. Mutation operates on the natural-language conventions that define a family (for in-context adversaries) rather than on model weights, keeping the red-team cost bounded. As a concrete illustration, a semantic encoding might tie the payload bit to whether a review leads with a strength or a weakness: both are truthful and common, yet across a batch the choice carries a reliable bit to a partner who knows the convention---which is why the feature set leans on cross-run mutual information against outcomes, not surface statistics alone.

\section{Estimation and Calibration Details}
\label{app:estimation}
Cross-run mutual information is estimated with a neural estimator~\cite{belghazi2018mine}; because such estimators are biased and high-variance in finite samples, we (i) calibrate decision thresholds on estimator \emph{outputs} over control episodes rather than on nominal bits (Algorithm~\ref{alg:detector}), and (ii) report bootstrap confidence intervals over calibration folds and seeds. We favour a neural estimator over plug-in or $k$-nearest-neighbour (KSG) estimators~\cite{kraskov2004estimating} because agent traces are high-dimensional and variable-length, a regime in which the latter degrade sharply. Because neural estimators such as MINE are high-variance and not strict bounds, we also compare against contrastive estimators---InfoNCE~\cite{oord2018representation} and the CLUB upper bound~\cite{cheng2020club}---whose trade-offs are analysed by Poole \emph{et al.}~\cite{poole2019variational}, validate every estimator on Gaussian data with known ground-truth MI before trusting it on traces, and adopt whichever has the lowest variance at matched sample size; as a further guard we cross-check KSG on low-dimensional projections and discard any configuration whose estimate has not stabilised as the sample size grows. The additive per-message decomposition assumed in Lemma~\ref{lem:partial} holds exactly when colluder messages are conditionally independent given the payload and the observed subset is drawn independently; correlated messages make the bound conservative, so $\epsilon_\rho\le\rho\epsilon$ is an upper estimate of what partial observation reveals.

\section{Evaluation Protocol}
\label{app:protocol}
Each configuration is evaluated over seeds as follows: instantiate an environment and population; run non-colluding populations to produce matched calibration episodes $\mathcal{C}$ and calibrate $t_\alpha$ on them (Algorithm~\ref{alg:detector}); run the co-evolution loop (Algorithm~\ref{alg:coevo}) for $R$ rounds, scoring only held-out encodings; and record the metrics of Section~\ref{sec:exp} under the partial-observability and multi-session variants. All baselines share the identical harness and calibration.

%=======================================================================
\bibliographystyle{IEEEtran}
\bibliography{references}

%=======================================================================

%\begin{IEEEbiography}[{\includegraphics[width=0.85in,height=1.0in,clip,keepaspectratio]{mcghanem-.jpg}}]{Mohamed Chahine Ghanem}
%(Member, IEEE) is an Associate Professor in Cyber Security at Keele University, U.K., and Co-founder of CyDRA Limited. He was Associate Professor at the University of Liverpool between 2022 and 2026. Dr Ghanem holds a PhD in Cyber Security Engineering from City, University of London, is a Senior Fellow of the HEA, and holds CISSP and multiple GIAC certifications. He is Editor-in-Chief of the \emph{Journal of Cyber Security and Risk Auditing}. Dr Ghanem has over 15 years of practice in senior positions in law enforcement and industry, including as global cyber-security auditing director for Kroll LLC and head of cyber-security auditing in banking.
%\end{IEEEbiography}

\end{document}